\documentclass[aps,pra,twocolumn,showpacs,superscriptaddress,groupedaddress]{revtex4-2}
\usepackage[utf8]{inputenc}
\usepackage{amsmath, amsthm, amssymb, amsfonts,enumitem}
\usepackage[makeroom]{cancel}
\usepackage{graphicx}
\usepackage{bm}
\usepackage{dsfont} 
\usepackage{mathtools}
\usepackage{booktabs}
\usepackage{thm-restate}
\usepackage[colorlinks=true, linkcolor=blue, urlcolor=blue,citecolor=blue]{hyperref}
\usepackage{orcidlink}
\usepackage{mathtools}
\DeclareMathOperator{\diag}{diag}
\DeclareMathOperator{\sign}{sign}
\DeclareMathOperator{\imm}{Imm}
\DeclareMathOperator{\per}{per}

\DeclareMathOperator{\tr}{tr}
\DeclareMathOperator{\SEP}{SEP}

\DeclareMathOperator{\id}{id}

\newcommand{\bra}[1]{\mathinner{\langle #1|}}
\newcommand{\ket}[1]{\mathinner{|#1\rangle}}

\newcommand{\dyad}[1]{| #1\rangle \langle #1|}
\newcommand{\ot}[0]{\otimes}
\newcommand{\one}[0]{\mathds{1}}

\newcommand{\cdn}{(\C^d)^{\otimes n}}

\newcommand{\R}{\mathds{R}}
\newcommand{\C}{\mathds{C}}
\newcommand{\WW}{\mathcal{W}}

\newcommand{\E}{\mathcal{E}}

\DeclareUnicodeCharacter{202C}{\^{i}}

\makeatletter
\newcommand{\vast}{\bBigg@{4}}
\newcommand{\Vast}{\bBigg@{5}}
\makeatother
\newtheorem{theorem}    {Theorem}
\newtheorem{cor}    {Corollary}
\newtheorem{proposition}[theorem]{Proposition}
\newtheorem{observation}[theorem]{Observation}

\usepackage{soul,xcolor}

\begin{document}

\title{
Entanglement and distillation from symmetric positive maps
}
\author{Albert Rico${}^{\orcidlink{0000-0001-8211-499X}}$}
\affiliation{
Physics department, Universitat Autònoma de Barcelona,
ES-08193 Bellaterra (Barcelona), Spain;}
\affiliation{
Faculty of Physics, Astronomy and Applied Computer Science, Institute of Theoretical Physics, Jagiellonian University,
30-348 Krak\'{o}w, 
Poland}
\date{\today}
\begin{abstract}
Recently, a toolkit of highly symmetric techniques employing matrix inequalities and positive maps has been developed to detect entanglement in various ways. Here we concisely review these methods providing a unified framework, and expand them to a new family of positive maps with further detection capabilities. 
In the simplest case, we generalize the reduction map to detect more generic states using both multiple copies and local filters. Through the Choi-Jamio{\l}kowski isomorphism, this family of maps leads to a construction of multipartite entanglement witnesses. Discussions and examples are provided regarding the detection of states with local positive partial transposition and the use of multiple copies.
\end{abstract}

\maketitle

\section{Introduction}
In recent years, the implementation of quantum protocols has gained refined control on mutlipartite and high-dimensional quantum systems~\cite{advancesHDent_Erhard2020,EfficientLargeScaleMBdyn_Artaco2024}. These advances are the groundwork for multiple applications like multipartite quantum networks~\cite{LargScaleQNetworksKozlowski_2019,NetworkGMPE_Navascues2020} and mid-scale quantum computation~\cite{MPEntSuperQubits_Lu2022}, where entanglement is crucially a major nonclassical resource. Therefore, the realistic implementation of quantum protocols requires characterizing and detecting the presence of entanglement in multipartite quantum systems~\cite{MPEntSuperQubits_Lu2022,DetMPentManyBody_Frerot2022}.  However, this is a challenging theoretical and experimental problem due to the growth of the state space~\cite{GURVITS2003sepNPhard,DetMBodyContinuous_Kunkel2022}. 

In search of effective solutions, a variety of methods displaying symmetries have been introduced both in theory and experiments~\cite{Neven_SymmetryResMomentsPT2021,Elben2020,DetMBodyPermuMoments_Liu2022,Huber2022DimFree}. These allow to better explore a desired corner of the space of positive maps and witnesses detecting entanglement. For this purpose, well-established fundamentals in matrix inequalities and the permutation group have been revisited~\cite{Neven_SymmetryResMomentsPT2021,Huber2021MatrixFO,PmapsWBAlg_Miqueleta2024}. In particular, increasing attention~\cite{MaassenSlides,Huber2020PositiveMA,Huber2021MatrixFO,TP_Rico24,PmapsWBAlg_Miqueleta2024} has been paid to a family matrix inequalities known as {\em immanant inequalities}~\cite{marshall1979inequalities}: linear combinations of permutations of products of matrix entries, that are nonnegative for all positive semidefinite matrices. A well known example is the determinant of a $k\times k$ matrix $G$,
\begin{equation}\label{eq:DetIntro}
    \det(G)=\sum_{\pi\in S_k}\sign(\pi)\prod_{i=1}^k G_{i,\pi^{-1}(i)}\,,
\end{equation}
where $S_k$ is the symmetric group and $\pi\in S_k$ are the possible permutations of $k$ elements. A simple immanant inequality states that the determinant is nonnegative, $\det(G)\geq 0$, for any positive semidefinite matrix $G\geq 0$. 

Immanant inequalities can be used to construct symmetric linear entanglement witnesses~\cite{MaassenSlides}. These can be grouped to construct local unitary-invariant multicopy witnesses that can be evaluated with randomized measurements~\cite{Elben22Toolbox}, when multiple copies of each local party evaluate an independent observable~\cite{TP_Rico24}. The construction of such witnesses employs trace polynomials~\cite{Huber2020PositiveMA,Huber2021MatrixFO,Huber2022DimFree}, which are homogeneous polynomials involving a set of matrices and their traces, and allow to construct positive maps invariant under the diagonal action of the unitary group. Therefore, trace polynomials and matrix inequalities together compose a toolkit that has been proven useful for entanglement detection and particularly suitable to systems with certain symmetries, but whose limitations are currently unknown.

Here we further develop this toolkit to new detection capabilities, for finite-dimensional Hilbert spaces. We present families of positive maps and witnesses constructed form projections onto subspaces with specific symmetries  (Proposition~\ref{prop:GeneralMap}), and show advantage over previous methods~\cite{Huber2020PositiveMA,Huber2021MatrixFO} in simple cases (Figure~\ref{fig:NPTabcTriangle}, Table~\ref{tab:DetectionRateAmap} and Observation~\ref{obs:ActivationReduction}). 
Using such maps, we construct a family of symmetric entanglement witnesses (Proposition~\ref{prop:WitPartContr}) as a special case of~\cite{short}. We expand on the results given in~\cite{short} in this regard: we compare the performance of witnesses constructed in this way with respect to their building blocks, in detecting states with local positive partial transpositions (Fig.~\ref{fig:ContrWit}, and Table~\ref{tab:PPTi}). Further examples detecting multiple copies are given.

The contributions are presented in the following scheme. In Section~\ref{sec:Preliminaries} we introduce the toolikt we use: we sketch how matrix inequalities and trace polynomials have been used in recent literature to detect entanglement.  In Section~\ref{sec:BipEntPmapsResults} we provide a new family of positive maps generalizing existing ones, which are equipped with distillation filters to better detect entanglement. In Section~\ref{sec:MultipEnt} we use the state-witness contraction method~\cite{short} to construct both linear and nonlinear multipartite witnesses using projectors onto selected subspaces.

\section{Preliminaries}\label{sec:Preliminaries}
\subsection{Entanglement detection: positive maps and witnesses}\label{subsec:ED,PM,WitsIntro}
A prominent tool to detect entanglement in a quantum state consists of positive maps $\mathcal{E}$, satisfying $\E(\varrho)\geq 0$ for all states $\varrho\geq 0$ such that bipartite states are detected through $\E_A\ot\id_B(\varrho_{AB})\not\geq 0$. It was shown in~\cite{horodecki2001separability} that a bipartite state $\varrho_{AB}$ is entangled if and only if there exists a positive map detecting it. Two examples of positive maps that will be used in this work are the transposition $\ket{i}\bra{j}^T=\ket{j}\bra{i}$ acting on one subsystem of a bipartite state $\varrho_{AB}$, known as the {\em partial transpose}~\cite{peres1997quantum}; and the {\em reduction map}~\cite{HoroRedCrit_1999}
\begin{equation}\label{eq:RedMap}
    R(\varrho)=\tr(\varrho)\one-\varrho
\end{equation}
acting on one subsystem of a bipartite state $\varrho_{AB}$, which is weaker than the partial transpose but ensures distillability of the detected states~\cite{HoroRedCrit_1999}. We introduce the reduction map because it will be naturally generalized later through the techniques introduced in this manuscript. Through the Choi-Jamio{\l}kowski isomorphism~\cite{choi1975completely,jamiolkowski1972linear}, each positive map that is not completely positive provides an entanglement witness $W$~\cite{GeoQuantStates2006} that detects entanglement if it has a negative expectation value, $\tr(W\varrho)<0$. We call $W$ as the {\em Jamio{\l}kowsk matrix} of the map it defines~\cite{GeoQuantStates2006}.

States that cannot be distilled to pure entangled states with local operations and classical communication (LOCC)~\cite{HorodeckiBoundEnt_1998} are known as {\em bound entangled}. This is the case of states with a positive partial transposition $T_A$ (PPT),  $\varrho^{T_A}\geq 0$. PPT states cannot be detected by so-called {\em decomposable} witnesses decomposing as $W = X + Y^{T_A}$ with $X,Y\geq 0$ positive semidefinite~\cite{OptEW_Lewenstein2000}.  Within the various ways these notions can be extended to the multipartite case~\cite{4partUnlockBound_Smolin2001,SepDistMultiPartBoundUnlock_Dur1999}, we consider the case of spatially separated local parties: we say that a $k$-partite state $\varrho$ is {\em locally bound entangled} if pure multipartite entanglement cannot be distilled when the local parties are restricted to LOCC, but joint operations are required~\cite{
4partUnlockBound_Smolin2001,
QSSMultiBoundEnt_Zhou2018,
SepDistMultiPartBoundUnlock_Dur1999,
BoundMaxViolBellUnlockApplic_Augusiak2006,
BellCorActBoundEnt_Bandy2005,
RemoteInfoConcUnlockBoundApplic_Murao2001,
QComComplexBoundUnlockApplic_Bruckner2002}. This occurs if all local partial transpositions are positive semidefinite.

\subsection{Immanants}\label{subsec:imms}
Let us introduce some notation we need concerning the group of permutations, which will be exemplified later in well-known cases. Consider the group $S_k$ of permutations $\pi\in S_k$ of $k$ elements, known as the symmetric group. A partition $\lambda=(\lambda_1,...,\lambda_{|\lambda|})\vdash k$ of a natural number $k$ satisfies $\lambda_1\geq...\geq\lambda_{|\lambda|}$ and $\lambda_1+...+\lambda_{|\lambda|}=k$, for example $(1,1)\vdash 2$ with $1+1=2$ and $(2,0)\vdash 2$ with $2+0=2$. Each partition $\lambda$ denotes an irreducible representation of the symmetric group $S_k$ and has a character $\chi_\lambda(\pi)$. Given a $k\times k$ matrix $G$, the {\em immanant} corresponding to the partition $\lambda$ reads
\begin{equation}\label{eq:immanant}
    \imm_\lambda(G) = \sum_{\pi \in S_k} \chi_\lambda(\pi) \prod_{i=1}^k G_{i \pi^{-1}(i)}\,.
\end{equation}
The product of diagonal elements $G_{ii}$ is obtained as a combination of all immanants of $G$ as $\prod_{i=1}^k G_{ii}=\sum_\lambda\chi_\lambda(\text{id})\imm_\lambda(G)/k!$, where the permutation $(\text{id})$ is the identity~\cite{marshall1979inequalities}. 

The extreme cases of immanants are well-known: 
\begin{itemize}
    \item The {\em determinant} $\det(G)=\imm_{1^k}(G)$ corresponds to the partition $\lambda=[1,1,...,1]:=[1^k]$ (associated with the antisymmetric representation of $S_k$), whose character $\chi$ is the sign of each permutation $\pi\in S_k$, $\chi(\pi)=\sign(\pi)$. The determinant is explicitly written in Eq.~\eqref{eq:DetIntro}.
    \item The {\em permanent} $\per(G)=\imm_{k}(G)$ corresponds to the partition $\lambda=[k,0,...,0]:=[k]$ (associated with the symmetric representation of $S_k$) whose character is trivial, $\chi(\pi)=1$ for all $\pi\in S_k$. The permanent is written as
    \begin{equation}\label{eq:permanent}
    \per(G) = \sum_{\pi \in S_k}  \prod_{i=1}^k G_{i \pi^{-1}(i)}\,.
\end{equation}
\end{itemize}
For $2\times 2$ matrices $G\geq 0$, only the permanent $\per(G)=G_{11}G_{22}+G_{12}G_{21}$ and determinant $\det(G)=G_{11}G_{22}-G_{12}G_{21}$ are defined since only  $(2,0)\vdash 2$ and $[1,1]\vdash 2$ partition the integer $2$. For $3 \times 3$ matrices one has three partitions: $[3,0,0]$, $[2,1,0]$ and $[1,1,1]$. The first and last partitions correspond to the permanent and the determinant. The partition $[2,1,0]$ is associated to the {\em standard} representation of the three-element permutation group $S_3$, and has characters $\chi(\id)=2$, $\chi(12)=\chi(23)=\chi(13)=0$ and $\chi(123)=\chi(132)=-1$. Therefore its corresponding immanant is given by
\begin{equation}
    \imm_{(2,1)}=2\prod_{i=1}^3G_{ii}-\prod_{i=1}^3G_{i,i+1}-\prod_{i=1}^3G_{i,i+2},
\end{equation}
where the index summation is modulo $3$.

\subsection{Immanant inequalities}\label{subseq:ImmIneqs}
Some inequalities involving immanants hold for any positive semidefinite matrix $G\geq 0$. These are known as {\em immanant inequalities}, have been extensively studied over years~\cite{bhatia2013matrix,marshall1979inequalities}, and yet still host open problems~\cite{Lieb_Permanent_Conj_1996}. Immanant inequalities have the form
\begin{equation}\label{eq:GenImmEq}
    \sum_{\lambda\vdash k}a_\lambda\imm_\lambda(G)\geq 0\quad\forall G\geq 0,
\end{equation}
with $a_\lambda\in\R$. For example, the determinant of a matrix $G$ in Eq.~\eqref{eq:DetIntro} equals the product of eigenvalues of $G$ and therefore is nonnegative for all $G\geq 0$. A less trivial example is the following~\cite{marshall1979inequalities}:
\begin{equation}\label{eq:HadamardIneq}
    \prod_{i=1}^k G_{ii}-\det(G)\geq 0
\end{equation}
for all $G\geq 0$, known as Hadamard inequality. 

In this work we will construct positive maps and witnesses combining immanants corresponding to different partitions $\lambda,\mu,...,\nu$ of an integer $k$. Therefore, we need to introduce a  shorthand notation to label each case. To do this, first note that each immanant inequality in Eq.~\eqref{eq:GenImmEq} is uniquely determined by its coefficients $a_\lambda$. Using this fact, we will identify each immanant inequality with a vector of real components,
\begin{equation}\label{eq:vecParts}
    \Vec{a}=(a_\lambda,a_\mu,...,a_\nu),
\end{equation}
where $\lambda,\mu,...,\nu$ are the possible partitions of an integer $k$ in a fixed order. Here this order will be fixed, using the convention that the components $\lambda_i$ of $\lambda=[\lambda_1,...,\lambda_{|\lambda|}]$ are in non-increasing order: when comparing two partitions $\lambda$ and $\mu$, let $\lambda_i$ be the first component of $\lambda$ such that $\lambda_i>\mu_i$. Then $\lambda$ preceeds $\mu$ if $\lambda_i>\mu_i$. If a partition $\lambda$ is not involved in the immanant inequality under consideration, its coefficient is $a_\lambda = 0$. 

Let us illustrate the notation above with an example. Recall that each vector $\Vec{a}$ in Eq.~\eqref{eq:vecParts} denotes an immanant inequality $\psi_{\vec{a}}(G)\geq 0$ defined by
\begin{equation}\label{eq:ImmNotation}
    \psi_{\vec{a}}(G):=a_\lambda\imm_\lambda(G)+...+a_\nu\imm_\nu(G).
\end{equation}
For example, consider the case $k=2$ with two possible partitions, $\lambda=[2,0]$ and $\mu=[1,1]$. 
For $2\times 2$ matrices $G$, the inequality $\imm_{[1,1]}(G):=\det(G)\geq 0$ is given by $a_{\lambda=[2,0]}=0$ and $a_{\mu=[1,1]}=1$. These coefficients are ordered as $a_{[2,0]}$ preceding $a_{[1,1]}$ because $\lambda_1=2>1=\mu_1$. Therefore they compose the vector $(a_{[2,0]},a_{[1,1]})=(0,+1)$ determining the following inequality:
\begin{align}
    \psi_{a_{[2,0]},a_{[1,1]}}(G):&=a_{[2,0]}\imm_{[2,0]}(G) + a_{[1,1]}\imm_{[1,1]}(G) \nonumber\\
    &=0\cdot \per(G) + 1\cdot\det(G)\nonumber\\
    &=\det(G)\geq 0\,.
\end{align}

\subsection{Entanglement witnesses from immanant inequalities}\label{subseq:EWfromImm}
Recently, a bridge has been found between immanant inequalities and the detection of entanglement~\cite{MaassenSlides,Huber2021MatrixFO,TP_Rico24}. The key idea is that given a $k\times k$ positive semidefinite matrix $G\geq 0$ with entries $G_{ij}=\bra{v_i}v_j\rangle$, where $\{\ket{v_i}\}_{i=1}^k$ is a set of any vectors, Eq.~\eqref{eq:immanant} can be written as
\begin{equation}\label{eq:immanant2}
    \imm_\lambda(G)=\frac{k!}{\chi_\lambda(\text{id})}\tr(P_\lambda\dyad{v_1}\otimes\cdots\otimes\dyad{v_k})\,.
\end{equation}
Here 
\begin{equation}\label{eq:YoungProjector}
P_\lambda=\frac{\chi_\lambda(\id)}{k!}\sum_{\pi\in S_k}\chi_\lambda(\pi)\eta_d(\pi^{-1})\,
\end{equation}
is the Young projector onto the irreducible subspace of ${(\C^d)}^{\otimes k}$ invariant under the irreducible representation $\lambda$ of the symmetric group and the unitary group, where $\eta_d(\pi)$ permutes $k$ tensor factors as
\begin{equation}
\eta_d(\pi)\ket{i_1} \ot \dots\ot \ket{i_k}
=
\ket{i_{\pi^{-1}(1)}} \ot \dots \ot \ket{i_{\pi^{-1}(k)}}\,.
\end{equation}
For example, $\eta_3(12)\ket{u}\otimes\ket{w}=\ket{w}\otimes\ket{u}$ permutes two qutrit states $\ket{u},\ket{w}\in\C^3$. For two qubits, all Young projectors are the symmetric (triplet) and antisymmetric (singlet) projectors, $P_{2,0}=\dyad{00}+\dyad{11}+\dyad{\psi^{+}}$ and $P_{1,1}=\dyad{\psi^-}$, with $\ket{\phi^\pm}=(\ket{01}\pm\ket{10})/\sqrt{2}$. The three-qutrit Young projectors are written in Appendix~\ref{App:YoungProjectors}.

Consider a separable state, $\varrho_{\SEP}=\sum_ip_i\dyad{v_1^i}\otimes\cdots\otimes\dyad{v_k^i}$. It follows from Eq.~\eqref{eq:immanant2} that for any immanant inequality $\sum_\lambda a_\lambda\imm_\lambda(X)\geq 0$ there exists an associated operator
\begin{equation}\label{eq:WitMaassen}
    R = \sum_{\lambda\vdash k}\frac{k!}{\chi_\lambda(\text{id})}a_\lambda P_\lambda
\end{equation}
that has nonnegative expectation value on all separable states, $\tr(R\varrho_{\text{SEP}})\geq 0$. These operators $R$ are entanglement witnesses if they have some negative eigenvalue, $R\not\geq 0$, or positive operators otherwise. For example, from the inequality~\eqref{eq:HadamardIneq} one can derive that
\begin{equation}
\tr(\eta_d(12)\varrho_{\text{SEP}})\geq 0\,,
\end{equation}
where the SWAP operator $\eta_d(12)$ is a witness detecting bipartite entanglement; and the inequality $\det(G)\geq 0$ on $2\times 2$ positive matrices $G\geq 0$ leads to $\tr(P_{1,1}\varrho_{\SEP})\geq 0$ which is trivial since the projector onto the two-qudit antisymmetric subspace
\begin{equation}\label{eq:2quditAntiSym}
    P_{1,1}=\frac{\one-\eta_d(12)}{2}
\end{equation}
is by definition positive semidefinite. When confusions may arise we shall adopt the following notation: if $R$ is a witness, it will be denoted by $W$; and if $R$ is a positive operator, it will be denoted as $M$. An operator defined from an immanant inequality with coefficients $\vec{a}=(a_\lambda,a_\mu,...,a_\nu)$ shall therefore be distinguished as $M_{\vec{a}}$ or $W_{\vec{a}}$. 

To end this explanation we highlight that, although  entanglement witnesses based on immanant inequalities are fundamental objects because Young projectors decompose the Hilbert space, their capabilities are currently at the beginning of their exploration. For instance, it is unknown for instance whether they can detect bound entanglement on bipartite states through their action on multiple copies.

\subsection{Positive maps from immanant inequalities}\label{eq:PMapsfromImm}
Trace polynomials are polynomial functions on input matrices and and their traces, for example $XY+\one\tr(XY)-X\tr(Y)$.
Using trace polynomials, the derivations in~\cite{lew1966generalized,Huber2020PositiveMA} introduced a family of maps from $k-1$ matrices $\varrho_i$ of size $d\times d$ to one matrix of size $d\times d$, acting as
\begin{align}\label{eq:TracePolyMapLambda}
    \Psi_\lambda(\varrho_1,...,\varrho_{k-1}):=\tr_{1,...,k-1}(P_\lambda( \varrho_1\otimes\cdots\otimes \varrho_{k-1}\otimes\one))\,.
\end{align}
The maps of this form are multilinear, namely linear in each entry $\varrho_i$. Using linear algebra and positivity of the Young projectors, $P_\lambda\geq 0$, one verifies that these maps are positive $\Psi_\lambda(\varrho_1,...,\varrho_{k-1})\geq 0$ on positive semidefinite matrices $\varrho_i\geq 0$~\cite{Huber2020PositiveMA}. 

Consider a collection of positive maps $\Psi_\lambda$,..., $\Psi_\nu$ constructed in this way from  immanants $\imm_\lambda$,..., $\imm_\nu$. Following the notation in Eq.~\eqref{eq:ImmNotation}, we will denote their linear combination as
\begin{equation}\label{eq:PosMapsNotation}
    \Psi_{(a_\lambda,...,a_\nu)}:=a_\lambda\Psi_{\lambda}+...+a_{\nu}\Psi_{\nu}\,.
\end{equation}
By reasoning analogously as above, we see that if $\psi_{\vec{a}}$ defines an immanant inequality, then $\Psi_{\vec{a}}$ defines a positive map~\cite{Huber2021MatrixFO}. These maps where shown in~\cite{Huber2021MatrixFO} to have several applications in quantum information theory. For example, this family of maps recovers well-known maps like the Cayley-Hamiltonian map~\cite{lew1966generalized} and the reduction map given in Eq.~\eqref{eq:RedMap}, through $\Psi_{(+1,-1)}(\varrho)=R(\varrho)$. Thus a straightforward application is the detection of entanglement, recovering for example the reduction criterion: if
\begin{equation}\label{eq:RedCrit}
    R\otimes\id(\varrho_{AB})\not\geq 0\,,
\end{equation}
then $\varrho_{AB}$ is entangled.

A further step was done in~\cite{TP_Rico24}. Collections of $n$ $k$-partite witnesses and positive operators of the form~\eqref{eq:WitMaassen} were evaluated onto $k$ copies of an $n$-partite state $\varrho$: each $k$-partite operator is evaluated onto $k$ local copies of each subsystem. Then a global nonlinear witness
\begin{equation}
    \WW = R_1\otimes...\otimes R_n
\end{equation}
detects entanglement in $\varrho^{\otimes k}$ through a trace polynomial expression that can be evaluated with randomized measurements~\cite{Elben2020,Elben22Toolbox}.

\section{Extension and combination of multilinear positive maps}\label{sec:BipEntPmapsResults}
\subsection{Extension with local degrees of freedom}\label{subsec:addLocDegofF}
Here we extend the families of multilinear positive maps and witnesses constructed in~\cite{Huber2020PositiveMA,Huber2021MatrixFO,MaassenSlides,TP_Rico24,PmapsWBAlg_Miqueleta2024} sketched above, as follows. 
\begin{proposition}\label{prop:GeneralMap}
Let the immanant inequality
\begin{equation}\label{eq:ImmEqPropMain}
\sum_{\lambda\vdash k}a_\lambda \imm_\lambda(G)\geq 0    
\end{equation}
hold for any $k\times k$ positive matrix $G\geq 0$, and let $E,\varrho_1,...,\varrho_{k-1}$ be positive matrices of size $d\times d$.   
The following multilinear map from $\C^{d^{k-1}}\times \C^{d^{k-1}}$ to $\C^d\times\C^d$,
\begin{align}\label{eq:PosMapMain}
\Psi_{\Vec{a}}^{E}(&\varrho_1\otimes...\otimes\varrho_{k-1}):=\\
&\sum_\lambda a_\lambda\frac{k!}{\chi_\lambda(\id)}\tr_{\setminus k}\Big (P_\lambda E^{\otimes k}(
\varrho_1\otimes ...\otimes \varrho_{k-1}\otimes\one)\Big ),\nonumber
\end{align}
is positive, namely $\Psi_{\Vec{a}}^{E}(\varrho_1\otimes...\otimes\varrho_{k-1})\geq 0$ when $\varrho_i\geq 0\,\,\forall\,\, 1\leq i\leq k-1$, where $\tr_{\setminus k}$ denotes partial trace over all subsystems except the last one.
\end{proposition}
This follows from positivity of $\Psi_{\vec{a}}$ in Eq.~\eqref{eq:PosMapsNotation}, as shown below (a less technical and more detailed proof is given in Appendix~\ref{app:ProofMainProp}). 
\begin{proof}
Consider the Hermitian and positive square root $\sqrt{E}$ of $E$, with $\sqrt{E}\sqrt{E}=E$. By Schur-Weyl duality, $\sqrt{E}^{\otimes k}$ commutes with the Young projectors $P_\lambda$ as these are in the center of the group algebra $\C S_k$. Together with the cyclic property of the trace, this implies the following, 
\begin{equation}
   \Psi_{\Vec{a}}^{E}(\varrho_1,...,\varrho_{k-1}) = \sqrt{E} \Psi_{\Vec{a}}(\tilde{\varrho}_1,...,\tilde{\varrho}_{k-1}) \sqrt{E},
\end{equation}
with $\tilde{\varrho}_i=\sqrt{E}\varrho_i\sqrt{E}$. This proves the claim, since $KMK^\dag\geq 0$ for any matrix $K$ if $M\geq 0$.
\end{proof}

As an example, the inequality $\det(G)\geq 0$ gives through Proposition~\ref{prop:GeneralMap} the following map using the $k$-qudit antisymmetrizer $P_{[1^k]}$,
\begin{align}\label{eq:ReductionA-map}
    \Psi_{(0,...,0,1)}^{E}(\{\varrho_i\})&=\tr_{1...k-1}\Big (P_{[1^k]}(E^{\otimes k})\bigotimes_{i=1}^{k-1}\varrho_i\otimes\one\Big )
\end{align}
for all $E,\varrho_i\geq 0$. For $k=2$ we have $P_{[1,1]}=(\one-\eta_d(12))/2$. Therefore, by setting $E=\one$ we recover the reduction map~\eqref{eq:RedMap}, $\Psi_{(0,1)}^{\one}(\varrho)=\one\tr(\varrho)-\varrho\geq 0$. By setting $\varrho=\one$ we have the degree-two polynomial 
$\Psi_{(0,1)}^{E}(\one)=E\tr(E)-E^2\geq 0$, which can be lifted to degree three as $\Psi_{(0,1)}^{E}(E)=E\tr(E^2)-E^3\geq 0$. We denote the positive matrices as $E$ and $\varrho_i$ respectively because in this work the input of a map $\Psi_{\vec{a}}^E$ will be a state to be detected (or a subsystem of it), and the positive matrix $E$ a tool to detect it. 
In general, a large variety of positive maps can be obtained from different immanant inequalities, values of $k$ and arbitrary positive operators $E\geq 0$. Examples following from the positivity of the determinant for $k=3$ are given in Appendix~\ref{app:ExsIneqsN=3}.

Note that the maps constructed through Proposition~\ref{prop:GeneralMap} are symmetric under permutations $\eta_d(\pi)$ with $\pi\in S_{k-1}$ acting on the input $k-1$ systems. This is because (1) the Young projectors $P_\lambda$ are in the center of the symmetric group algebra $\C S_k$ and thus commute with all permutations; and (2) by Schur-Weyl duality the diagonal actions of the symmetric group and the general linear group commute, leading to $E^{\otimes k-1}\eta_d(\pi)=\eta_d(\pi)E^{\otimes k-1}$. This also shows that the set of images of these maps is closed under the diagonal action of the general linear group: for any map $\Psi_{\vec{a}}^E$, if all $\varrho_i$ undergo the same probabilistic operation $K\varrho_i K^\dag$, then there is a map $\Psi_{\vec{a}}^{E'}$ with $E'={K^{-1}}^\dag E K^{-1}$ recovering the image of $\Psi_{\vec{a}}^E$.



    \subsection{Detection of random states}\label{subsec:DetRandomStates}
    It follows from Proposition~\ref{prop:GeneralMap} that entangled $(k-1)$-partite states $\varrho$ are detected if $\Psi_{\vec{a}}^E(\varrho)\not\geq 0$. Here we demonstrate the power of this method in detecting a single copy of random states, by choosing $E$ appropriately. 
    \begin{figure}[tbp]
        \centering
        \includegraphics[scale=0.6]{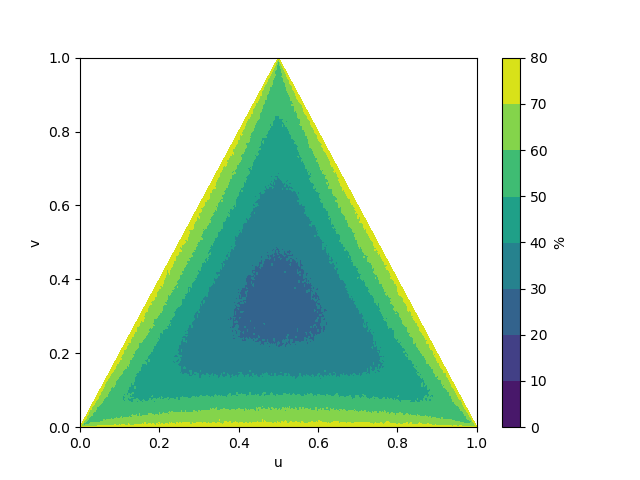}
        \caption{{\bf Detection of two-qutrit random states} by $\Psi_{(0,1)}^{E}\otimes\id(\varrho_{AB})\not\geq 0$ with $E=\diag(a,b,c)$. The axes parametrize $a$, $b$ and $c$ as $u:=(a-b+1)/2$ and $v:=c$ in order to plot the results in the probability simplex of size $3$, and the color bar indicates the percentage of Ginibre random bipartite states~\cite{ginibre1965} detected with in a sample of size 2000. At the baricenter of the triangle lays the case $E=\diag(1,1,1)=\one$, which evaluates the reduction map~\eqref{eq:RedMap} and detects $27\%$ of the random states approximately. The edges contain the case $E=\diag(1,1,0)$, which evaluates the reduction map on a two-dimensional subspace~\eqref{eq:RedMapProj2dim} and increases the detection rate to $74\%$. 
        }
        \label{fig:NPTabcTriangle}
    \end{figure}

    Let us now focus on the case where $E=P$ is a two-dimensional local projector, $P=\diag(1,1,0...,0)\otimes\one$. Then $\Psi_{(0,1)}^{P}(\varrho)$ acting on two-qudit states is equivalent to the reduction map acting on the projected matrix $P\varrho P$, namely
    \begin{equation}\label{eq:RedMapProj2dim}
        \Psi_{(0,1)}^{P}(\varrho) = R(P\varrho P)\,.
    \end{equation}
     Figure~\ref{fig:NPTabcTriangle} and Table~\ref{tab:DetectionRateAmap} summarize the detection with either $\Psi_{(0,1)}^{E}\otimes\id(\varrho_{AB})$ or $\id\otimes\Psi_{(0,1)}^{E}(\varrho_{AB})$ on a random sample of Ginibre random states $\varrho_{AB}$. In Figure~\ref{fig:NPTabcTriangle} we use diagonal operators $E$ with different coefficients. In Table~\ref{tab:DetectionRateAmap} we use an operator $E:=P=P_A\otimes\one$ with $\diag(1,1,0,...,0)$, namely a 2-dimensional projector. The superscript $P_A$ means that only party $A$ is projected, and the superscript $P_{A(B)}$ means that either party $A$ or $B$ are projected. The superscript $P_A^\pi$ means that detection of random states is tested for all possible permutations $\pi$ of diagonal entries in $P_A^\pi=\diag(\pi(1,1,0,...,0))$, in order to detect states with two-dimensional entanglement in any possible qubit-qudit subspace. The superscript $P_A^{U,\pi}$ means that the latter is tested for a sample of random unitaries $U$, to avoid restriction onto diagonal projectors in the computational basis. The results of this analysis lead to the following observations.
    \begin{observation}\label{obs:ActivationReduction}
    For local dimensions $2\leq d_A=d_B\leq 6$, most entangled Ginibre random states are not detected by the reduction map acting on a whole subsystem (last row),
    \begin{equation}\label{eq:RedMapNoDet}
        R_A\otimes\id_B(\varrho_{AB}) \geq 0\,;
    \end{equation}
    but can be detected after a local projection $P_A$ enabling distillation (all but last row),
    \begin{equation}\label{eq:RedMapProjAct}
        R_A\otimes\id_B\Big ((P_A\otimes\one_B)\varrho_{AB}(P_A\otimes\one_B)\Big ) \not\geq 0\,.
    \end{equation}
    \end{observation}
    In essence, here entanglement between qudit pairs is detected as qubit-qudit entanglement. Note that the detection rate of generic states with this method is comparable to that of the PPT-criterion~\cite{HoroRedCrit_1999,ReviewQEnt_Horo2009} (the second row of Table~\ref{tab:DetectionRateAmap}). The number of two-dimensional subspaces where to project scales with the Newton binomial $d!/2!(d-2)!$ with respect to the local dimension $d$. Comparatively, the rate of states detected such subspaces is significantly higher than that of the standard reduction map, which tends to zero when local dimensions grow. Since violating the reduction criterion implies distillability~\cite{HoroRedCrit_1999,clarisse2006PhDDistillation}, qubit-qudit entanglement distillation becomes accessible from qudit-qudit states detected in this way, where the projector $P$ acts as a filtering operation.
    \begin{table}[h]
        \centering
        \begin{tabular}{l c c c c c c}
        \hline
        
           {\bf Detection}  & \textbf{$d=2$} & \textbf{$d=3$} & \textbf{$d=4$} & \textbf{$d=5$} & \textbf{$d=6$} & Distillable \\
           \hline

            PPT-criterion & $7566$ & $9999$ & $10000$ & $10000$ & $10000$ & ?~\cite{HorodeckiBoundEnt_1998}  \\
             $\Psi_{(0,1)}^{P_{A(B)}^{U,\pi}}$  & $7566$ & $9998$ & $10000$ & $10000$ & $10000$ & Yes \\$\Psi_{(0,1)}^{P_{A(B)}^\pi}$  & $7566$ & $9924$ & $9991$ & $9988$ & $9869$ & Yes \\
          $\Psi_{(0,1)}^{P_{A(B)}}$ & $7566$ & $9016$ & $8259$ & $6231$ & $3172$ & Yes \\
            $\Psi_{(0,1)}^{P_A}$~\eqref{eq:RedMapProj2dim}  & $7566$ & $7406$ & $6078$ & $3874$ & $1732$ & Yes  \\
           $\Psi_{(0,1)}^\one$~\eqref{eq:RedMap} & $7566$ & $2741$ & $10$ & $0$ & $0$ & Yes~\cite{HoroRedCrit_1999} \\

           \hline
        \end{tabular}
        \caption{{\bf Number of generic bipartite detected states with $\Psi_{(0,1)}^{P_A}$.} A set of $10^4$ Ginibre random bipartite states $\varrho_{AB}$ with local dimension $d_A=d_B:=d$ is sampled. We consider the number of states detectd by the PPT-criterion~\cite{GeoQuantStates2006} ({\em first row});  
        by either $\Psi_{(0,1)}^{P_A}\otimes\id(\varrho_{AB})\not\geq 0$ or $\id\otimes\Psi_{(0,1)}^{P_B}(\varrho_{AB})\not\geq 0$ with $P=\diag(1,1,0,...,0)$, considering $50d$ Haar-random unitaries $U_A\otimes U_B$ and all permutations $\pi$ acting on $P_{A(B)}$ -- denoted by $P_{A(B)}^{U,\pi}$ ({\em second row}); 
        by either $\Psi_{(0,1)}^{P_A}\otimes\id(\varrho_{AB})\not\geq 0$ or $\id\otimes\Psi_{(0,1)}^{P_B}(\varrho_{AB})\not\geq 0$, considering all permutations $\pi$ on $P_{A(B)}$ -- denoted by $P_{A(B)}^\pi$ ({\em third row}); 
        by either $\Psi_{(0,1)}^{P_A}\otimes\id(\varrho_{AB})\not\geq 0$ or $\id\otimes\Psi_{(0,1)}^{P_B}(\varrho_{AB})\not\geq 0$ ({\em fourth row}); 
        by $\Psi_{(0,1)}^{P_A}\otimes\id(\varrho_{AB})\not\geq 0$ ({\em fifth row}); 
        and by the reduction map~\eqref{eq:RedMap} ({\em sixth row}).}
        \label{tab:DetectionRateAmap}
    \end{table}



\begin{figure*}[tbp]
    \centering
    \includegraphics[width=1.\linewidth]{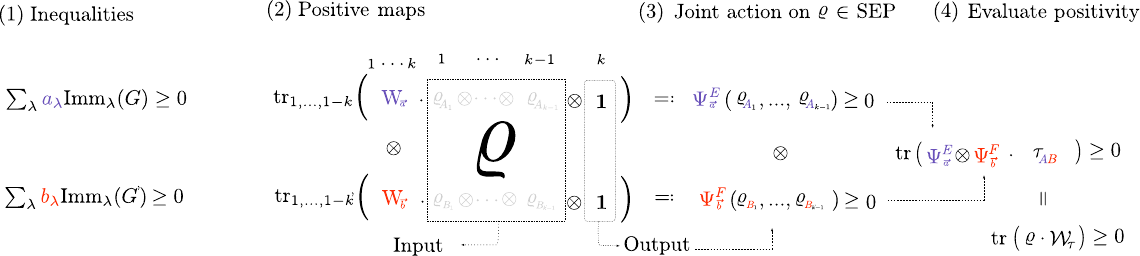}
    \caption{{\bf Entanglement detection scheme proposed in this work}. The method employed here works in the following steps for $n=2$ matrix inequalities, whose extensions to an arbitrary number $n$ are straightforward. (1) Chose two immanant inequalities with coefficients $\{a_\lambda\}$ and $\{b_\lambda\}$ given by $\vec{a}$ and $\vec{b}$, acting on matrices of size $k$ and $k'$. (2) Construct the corresponding $k$- and $k'$-partite witnesses or positive operators $W_{\vec{a}}=E^{\otimes k}\sum_\lambda a_\lambda P_\lambda$ and $W_{\vec{b}}=F^{\otimes k'}\sum_\lambda b_\lambda P_\lambda$, in the spirit of~\cite{MaassenSlides}. Following Eq.~\eqref{eq:CombineMaps}, these will be used as Jamio{\l}kowski matrices~\cite{jamiolkowski1972linear} of positive maps $\Psi_{\vec{a}}:\mathcal{L}(\C^d)^{\otimes k-1}\rightarrow \mathcal{L}(\C^d)$ and $\Psi_{\vec{b}}:\mathcal{L}(\C^d)^{\otimes k'-1}\rightarrow \mathcal{L}(\C^d)$ constructed in Proposition~\ref{prop:GeneralMap}, where we set $k\geq k'$. In their joint input space we place a $(k+k'-2)$-partite quantum state $\varrho$ to be detected. (3) If such a state is separable (a convex combination of product states $\varrho_{A_1}\otimes...\otimes\varrho_{A_{k-1}}\otimes\varrho_{B_1}\otimes...\otimes\varrho_{B_{k'-1}}$), the joint action of the two maps outputs a product of positive matrices $\Psi_{\vec{a}}\otimes\Psi_{\vec{b}}$. (4) Positivity of the output is then evaluated with a $(n=2)$-partite state or witness $\tau$, that detects entanglement in $\varrho$ if $\tr(\Psi_{\vec{a}}\otimes\Psi_{\vec{b}}\cdot\tau)<0$. In the dual formulation of this procedure, $\WW_\tau=\tr_{A_kB_{k'}}(W_{\vec{a}}\otimes W_{\vec{b}}(\one\otimes\tau))$ is a witness for $\varrho$ ({\em state-witness contraction} technique of~\cite{short}).}
    \label{fig:GeneralFramework}
\end{figure*}

\subsection{Combination of multilinear maps}
Here we will evaluate a multilinear (in fact bilinear) map on one party of two copies of a bipartite state $\varrho_{AB}$. Conversely, one can combine several multilinear maps to detect a larger entangled state, such that each of the maps acts on certain subsystems. That is, consider $n$ multilinear maps $\Psi_{\vec{a_j}}^{E_j}$ of the form~\eqref{eq:PosMapMain} acting on $k-1$ systems. Then a $n\times(k-1)$-partite state $\varrho$ can be detected if
\begin{equation}\label{eq:CombineMaps}
    \Psi_{\vec{a_1}}^{E_1}\otimes\dots\otimes \Psi_{\vec{a_n}}^{E_n}(\varrho)\not\geq 0,
\end{equation}
where each map $\Psi_{\vec{a_j}}^{E_j}$ acts on the subsystems $(j-1)(k-1)+1,...,j(k-1)$. Thus $\Psi_{\vec{a_1}}^{E_1}\otimes\dots\otimes \Psi_{\vec{a_n}}^{E_n}$ can be seen as a multilinear map acting on $n(k-1)$ subsystems constructed from those acting in its subsets. This technique is depicted in Fig.~\ref{fig:GeneralFramework} and will be the main focus of the coming sections, which provide detailed insight and expand upon the results of~\cite{short}. 
Note that we considered all multilinear maps to act on the same number of subsystems $k-1$, but the technique can be extended to combining multilinear maps with different input sizes $k_j-1$ to detect states shared among $k_1+...+k_n-n$ local parties.

For example, consider the positive map of Eq.~\eqref{eq:ReductionA-map} defined from the inequality $\det(G)\geq 0$. In particular, let us consider its action onto $k-1$ copies of a bipartite state $\varrho_{AB}$ (Fig.~\ref{fig:RedCritKcopies}), 
\begin{align}\label{eq:REdmapkcopies}
\Psi^{\det}_A\big (\varrho_{AB}^{\otimes {k-1}})&:={\Psi_{(0,...,0,1)}^{\one}}_{A}\otimes\id_{B}(\varrho_{AB}^{\otimes {k-1}})\\
&=\tr_{A_{1.k-1}}({P_{[1^k]}}_{A_{1.k}}\otimes\one_{B_{1.k-1}}\,\varrho_{AB}^{\otimes k-1}\otimes \one_{A_k}\big ),\nonumber
\end{align}
where we define ${A_{1.k-1(k)}}:=A_1...A_{k-1}(A_{k})$ and ${B_{1.k-1}}:=B_1...B_{k-1}$. It is clear that if $\varrho_{AB}$ is separable, then
\begin{equation}\label{eq:REdCritkcopies}
\Psi^{\det}_A(\varrho_{AB}^{\otimes {k-1}})\geq 0\,,
\end{equation}
and thus 
Notice that the reduction criterion~\eqref{eq:RedCrit} is obtained by~\eqref{eq:REdCritkcopies} for $k=2$, namely when using a single copy of $\varrho_{AB}$ and the two-qudit antisymmetrizer~\eqref{eq:2quditAntiSym}.
\begin{figure}[tbp]
    \centering
    \includegraphics[width=1\linewidth]{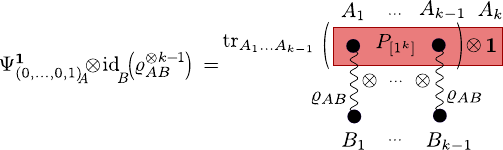}
    \caption{{\bf Nonlinear entanglement detection} with the action of $\Psi^{\one}_{(0,...,0,1)}$ on the party $A$ of $k$ copies of a bipartite state $\varrho_{AB}$, namely Eq.~\eqref{eq:REdmapkcopies}. A state is detected if the resulting operator acting on $B_1...B_{k-1}A_k$ is not positive semidefinite, see Eq.~\eqref{eq:REdCritkcopies}. The case $k=2$ recovers the (single-copy) reduction criterion~\eqref{eq:RedCrit}, whose detection capabilities on generic states are significantly improved by the (two-copy) criterion obtained for $k=3$ (Obs.~\ref{obs:MultCopyGenRMapGinibre}).}
    \label{fig:RedCritKcopies}
\end{figure}

\begin{observation}\label{obs:MultCopyGenRMapGinibre}
    In a sample of 1000 two-qutrit Ginibre random states~\cite{ginibre1965}, the reduction criterion~\eqref{eq:RedCrit}~\cite{HoroRedCrit_1999} can detect 249 of them; but 979 are detected with two copies as Eq.~\eqref{eq:REdCritkcopies} (see Fig.~\ref{fig:RedCritKcopies}) with $k=3$ is not satisfied.
\end{observation}

\section{New witnesses from state-witness contractions}\label{sec:MultipEnt}
\subsection{New multipartite witnesses}\label{subsec:NewWitnesses}
We have seen that Proposition~\ref{prop:GeneralMap} offers a variety of methods to detect entanglement. Through the Choi-Jamio{\l}kowski isomorphism, it yields entanglement witnesses of the form
\begin{equation}\label{eq:WitMaassenExt}
 W = E^{\otimes n}\sum_{\lambda\vdash n}\frac{a_\lambda}{\chi_\lambda(\id)}P_{\lambda},
\end{equation}
where $\chi_\lambda(\id)$ is the character of the identity in the partition $\lambda$ and $P_\lambda$ the Young projector onto a subspace with specific symmetries, generalizing the construction of~\cite{MaassenSlides} (see Appendix~\ref{app:GenMultiMapsWits}). Using positivity of Eq.~\eqref{eq:CombineMaps}, these can be combined using state-witness contraction~\cite{short} (Fig.~\ref{fig:ContrWit}):
\begin{proposition}\label{prop:WitPartContr}
 Let $\{W^{(i)}\}_{i=1}^n$ be $k$-partite witnesses or positive operators constructed from immanant inequalities as in Eq.~\eqref{eq:WitMaassenExt}. Let $\tau$ be a $n$-partite state or witness in the subsystem $\mathcal{S}=\{k,2k,...,nk\}$. Then, the operator
 \begin{equation}\label{eq:WitGenkn}
  \WW_\tau := \tr_{\mathcal{S}}\Big ( W^{(1)}_{1,...,k}\otimes...\otimes W^{(n)}_{k(n-1),...,nk}\,\cdot\,\one\otimes\tau_{\mathcal{S}} \Big )
 \end{equation}
has a nonnegative expectation value on $n\times(k-1)$-partite separable states, namely
 \begin{equation}
  \tr(\WW_\tau\varrho)\geq 0
 \end{equation}
if $\varrho$ is fully separable. 
\end{proposition}
This can be seen from positivity of Eq.~\eqref{eq:CombineMaps} using the Choi-Jamio{\l}kowski isomorphism, and the proof is detailed in Appendix~\ref{app:ProofPropWitnesses}.
One can think of {\em contracting} the $k$'th party of all $n$ witnesses (or positive operators) with the $n$-partite state or witness $\tau$, to create a new witness in the $(k-1)\times n$ non-contracted subsystems~\cite{short}. A specific example is given in Fig.~\ref{fig:ContrWit}. Similarly as in Eq.~\eqref{eq:CombineMaps}, extensions using witnesses acting on different number of subsystems, $W^{(1)}_{1,...,k},...,W^{(n)}_{1,...,k'}$, as depicted in Fig.~\ref{fig:GeneralFramework}, are straightforward and can detect $(k-1)+...+ (k'-1)$-partite states. 

By construction, witnesses constructed in this way are highly symmetric. Since the Young projectors $P_\lambda$ with $\lambda\vdash (k-1)$ are in the center of $\C S_{k-1}$ (the group algebra generated by linear combinations of permutations in the Hilbert space representation), each witness $\WW_\tau$ is invariant under the action of $S_{k-1}^{\times n}$. That is,
\begin{equation}
    \eta_d(\pi)\otimes...\otimes\eta_d(\upsilon)\,\WW_\tau\,\eta_d(\pi^{-1})\otimes...\otimes\eta_d(\upsilon^{-1}) = \WW_\tau
\end{equation}
where $\pi\in S_{k-1},...,\upsilon\in S_{k-1}$ are $n$ permutations of $k-1$ elements. Moreover, since Young projectors are linear combinations of permutations, $P_\lambda\in\C S_{k-1}$, Schur-Weyl duality implies that $\WW_\tau$ is invariant under the block-diagonal action of the unitary group,
\begin{equation}
    (U\otimes...\otimes V)^{\otimes k-1}\,\WW_\tau\,{(U^{\dag}}\otimes...\otimes {V^{\dag}})^{\otimes k-1}=\WW_\tau\,,
\end{equation}
where $U,...,V\in\mathcal{U}(d)$ are $n$ unitary matrices.

\begin{table}[tbp]
    \centering
    \begin{tabular}{l c c c c c}
    \hline
        {\bf $W$} & $3!P_{3}-\one$ & $P_{3}-P_{2,1}/4$ & $P_{2,1}/4-P_{1,1,1}$ & $\one-3!P_{1,1,1}$ \\
        \hline
       $d=2$  & $1$ & $1$ & PSD & PSD \\
       $d=3$ & $1$ &  $1$ & $2$ & $1$  \\
       $d=4$ & $1$ & $1$ & $2$ & $1$   \\
       $d=5$ & $1$ & $1$ & $2$ & $1$\\  
       \hline
    \end{tabular}
    \caption{{
   \bf Maximum number $0\leq t\leq 3$ such that} $\min\{\tr(XW):X,X^{T_1},...,X^{T_t}\geq 0\}<0$. By symmetry of the Young projectors $P_{\lambda}$, it is enough to impose positivity of the partial transpose in the first $t$ subsystems. Here PSD denotes that $W$ is a positive semidefinite operator, and not a witness, which occurs because $P_\lambda=0$ if $|\lambda|>d$.}
    \label{tab:PPTi}
\end{table}
\begin{figure*}[]
    \centering
    \includegraphics[width=1\linewidth]{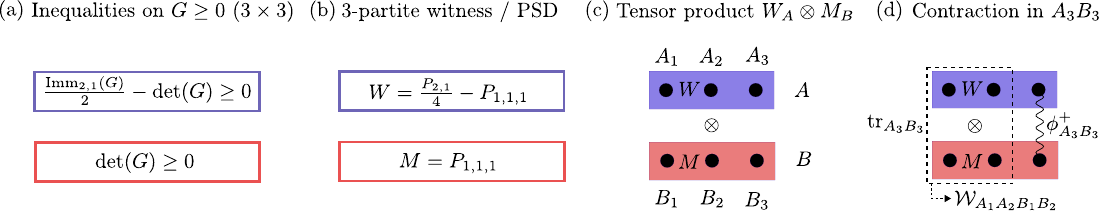}
    \caption{{\bf Construction of a four-partite witness from two matrix inequalities.} (a) Consider two immanant inequalities on $3\times 3$ positive semidefinite matrices $G$. (b) Each inequality can be used to construct a three-partite witness $W$ or positive operator $M$~\cite{MaassenSlides}. 
    (c) Take the tensor product $W_A\otimes M_B$, which can be used for example as a witness on three copies $\varrho_{A_1B_1}\otimes\varrho_{A_2B_2}\otimes\varrho_{A_3B_3}$~\cite{TP_Rico24}. (d) Here we contract the last subsystem with a Bell state $\phi^+=\dyad{\phi^+}$ shared between parties $A_3$ and $B_3$, as displayed in Eqs.~\eqref{eq:WitGenkn} and~\eqref{eq:Wit4partNonDeco}. Proposition~\ref{prop:WitPartContr} shows that we are left with a four-partite entanglement witness $\WW_\tau$ in the subsystems $A_1A_2B_1B_2$, and Observation~\ref{obs:4PartWitDetBoundEnt} shows that it can detect states with local positive partial transpositions.}
    \label{fig:ContrWit}
\end{figure*}

\subsection{Detection of states with local positive partial transpositions}\label{subsec:WitBoundEntanglement}
Here we will address the question of whether witnesses constructed from immanant inequalities can detect bound entanglement, expanding on the results in~\cite{short}. For known techniques using immanant inequalities, the answer is either unknown~\cite{Huber2021MatrixFO,Huber2022DimFree,MaassenSlides,PmapsWBAlg_Miqueleta2024,TP_Rico24} or negative~\cite{Huber2020PositiveMA}. Even using broader classes of trace polynomial techniques evaluated through randomized measurements~\cite{Elben22Toolbox}, most existing results are focused on evaluating purities and moments of the partial transpose~\cite{Elben2020,Neven_SymmetryResMomentsPT2021}, which are suitable to highly entangled states; and most results in detecting bound entanglement with trace polynomials apply to bipartite~\cite{ImaiBoundfromRM_2021,liu2022MomentsPermutations,zhang2023experimentalverificationboundmultiparticle} or multiqubit systems~\cite{CollectiveRM_Imai2024}, and rely on approximating established criteria from copies of the state in hand. 

As a brief analysis, Table~\ref{tab:PPTi} shows that three-partite states with positive local partial transpositions cannot be detected by witnesses of the form ~\eqref{eq:WitMaassen} in~\cite{MaassenSlides} constructed from well-known immanant inequalities. However, Proposition ~\ref{prop:WitPartContr} suggests a reuse of such witnesses: states with two positive partial transpositions can be detected (Table~\ref{tab:PPTi}), and therefore witness-state contraction~\cite{short} might give rise to a larger witness detecting states with local positive partial transpositions.  Indeed, consider the witness
\begin{equation}\label{eq:Wit2deco}
 W = \frac{P_{2,1}}{4}-P_{1,1,1}\,,
\end{equation}
where the Young projectors are explicitly given in Appendix~\ref{App:YoungProjectors}, and the positive operator $P_{1,1,1}$, constructed from the inequalities $\imm_{2,1}(G)/2-\det(G)\geq 0$ and $\det(G)\geq 0$ respectively. For simplicity, we will use the shorthand notation $A_{12}:=A_1A_2$ and $A_{123}:=A_1A_2A_3$ ($B_{12}:=B_1B_2$ and $B_{123}:=B_1B_2B_3$) to concatenate local subsystems $A_i$ ($B_i$).
\begin{observation}[~\cite{short}]\label{obs:4PartWitDetBoundEnt}
Let $\ket{\phi^+}=\sum_{i=0}^{2}\ket{ii}/\sqrt{3}$ be the two-qutrit Bell state, and choose $\tau=\dyad{\phi^+}$; let $W$ be the witness of Eq.~\eqref{eq:Wit2deco}, and let $P:=P_{1,1,1}$ be the three-qutrit antisymmetrizer. The four-qutrit entanglement witness
 \begin{equation}\label{eq:Wit4partNonDeco}
  \WW_{\tau} = \tr_{A_3B_3}\Big (W_{A_{123}}\otimes P_{B_{123}}\cdot\one_{A_{12}B_{12}}\otimes\tau_{A_3B_3}\Big )
 \end{equation}
 can detect four-partite entangled states $\varrho_{A_{12}B_{12}}$ with $\varrho^{T_{A_1}},\varrho^{T_{A_2}},\varrho^{T_{B_1}},\varrho^{T_{B_2}}\geq 0$ through
 \begin{equation}
  \tr(\WW_{\tau}\varrho)<0\,.
 \end{equation}
\end{observation}

\subsection{Detection through multicopy witnesses}
Here we will shortly expand the results of state-witness contraction~\cite{short} in detecting a state thorugh multiple copies. That is, the multipartite case discussed above applies in particular when considering $k-1$ copies of a $n$-partite state, $\varrho^{\otimes k-1}$, as the state to be detected (e.g. in Fig.~\ref{fig:GeneralFramework}). Consider $n$ positive maps derived through Proposition~\ref{prop:GeneralMap} from $n$ Immanant inequalities,  $\Psi_{\Vec{\alpha_l}}^{X^{(l)}}:=\Psi_l$. 
It follows from Proposition~\ref{prop:WitPartContr} that for separable states we have
\begin{equation}\label{eq:ThetaGeneral}
\theta_\tau:=\tr\big (\Psi_1\otimes...\otimes\Psi_n(\varrho^{\otimes k-1})\tau\big )\geq 0\,,
\end{equation}
whose violation is a criterion to detect entanglement in $\varrho$. As a trivial example, consider $\Psi^{\one}_{(0,+1)}$ acting on $k-1=1$ copy of $A$ and $\Psi^{E}_{(+1,-1)}$ acting on one copy of the party $B$. 
We obtain that inequalities like
\begin{equation}
    \tr\Big (E\varrho_BE\tau_B\Big ) - \tr\Big ((\one\otimes E)\varrho_{AB}(\one\otimes E)\tau_{AB}\Big )\geq 0
\end{equation}
hold for any state $\tau_{AB}$ and positive operator $E$ if $\varrho_{AB}$ is separable. In particular for $\tau=\varrho$ and $E=\one$ we recover the purity inequality $\tr(\varrho_B^2)-\tr(\varrho_{AB}^2)\geq 0$~\cite{EntDetRev_Guhne2009}.

By setting $E=\one$ in all positive maps $\Psi_l$ and $\tau=\varrho$ we recover the family of witnesses acting on $k$ copies of $\varrho$ that we introduced in~\cite{TP_Rico24}, which for the bipartite case can be written as
\begin{equation}\label{eq:WitTP-ent}
    \WW = \big (\sum_\lambda a_\lambda P_\lambda\big )_{A_1...A_k}\otimes\big (\sum_\mu b_\mu P_\mu\big )_{B_1...B_k}\,.
\end{equation}
Expressions obtained in this way evaluate $k$'th moments of the state $\varrho$ and its reductions, generalizing purity inequalities. These expressions are invariant under local unitaries and implementable with randomized measurements~\cite{Elben2020,Elben22Toolbox}, as shown in~\cite{TP_Rico24}.

\section{Conclusions}
We have introduced a family of positive maps from matrix inequalities, able to detect entanglement in a variety of configurations. Standard methods like the reduction map or purity inequalities, as well as recently introduced more powerful techniques, are recovered as special cases. We analyze the simplest example, where significant advantage over the reduction map is shown in a more general family. This shows that the detection of entangled states can be considerably improved by projecting locally on low-dimensional subspaces. We also provide further details on how the framework introduced in~\cite{short} allows to find the first known combination of immanant inequalities that can detect entangled states whose local partial transpositions are positive. The methods discussed explore a highly symmetric corner of the space of entanglement witnesses, whose limitations are at the moment unknown. For future work, it would be interesting to explore possible combinations of the methods introduced and other methods such as the symmetric extension hierarchy~\cite{doherty2004complete}. 

The code and data used for this project are openly available at~\cite{code-data}.

\bigskip

{\em Acknowledgements:\,\,} The author is thankful to 
Gerard Anglès Munné, 
Matt Hoogsteder-Riera, 
Pawe\l{} Horodecki, 
Felix Huber, 
Robin Krebs, 
Some Sankar Bhattacharya, 
Anna Sanpera, 
Andreas Winter 
and 
Karol \.{Z}yczkowski 
for discussions. Special thanks to Felix Huber for pointing out a first example of Proposition~\ref{prop:GeneralMap} and for editorial suggestions, and to Karol \.{Z}yczkowski for suggesting a symmetric parametrization in Fig.~\ref{fig:NPTabcTriangle}.  Support by the Foundation for Polish Science through TEAM-NET project  POIR.04.04.00-00-17C1/18-00
and by NCN QuantERA
Project No. 2021/03/Y/ST2/00193 is gratefully acknowledged. The author also acknowledges financial support from Spanish MICIN (projects: PID2022:141283NBI00;139099NBI00) with the support of FEDER funds, the Spanish Goverment with funding
from European Union NextGenerationEU (PRTR-C17.I1), the Generalitat de Catalunya,
the Ministry for Digital Transformation and of Civil Service of the Spanish Government through the QUANTUM ENIA project -Quantum Spain Project- through the Recovery, Transformation and Resilience Plan NextGeneration EU within the framework
of the Digital Spain 2026 Agenda.  

\addcontentsline{toc}{subsection}{Bibliography}
\bibliographystyle{ieeetr}
\bibliography{Bibliography}{}

@article{TP_Rico24,
  title = {Entanglement Detection with Trace Polynomials},
  author = {Rico, Albert and Huber, Felix},
  journal = {Phys. Rev. Lett.},
  volume = {132},
  issue = {7},
  pages = {070202},
  numpages = {6},
  year = {2024},
  month = {Feb},
  publisher = {American Physical Society},
  doi = {10.1103/PhysRevLett.132.070202},
  url = {https://link.aps.org/doi/10.1103/PhysRevLett.132.070202}
}

@article{short,
      title={State-witness contraction}, 
      author={Albert Rico},
      year={2025},
      journal={arXiv:2502.17697}
}

@article{ReviewQEnt_Horo2009,
  title = {Quantum entanglement},
  author = {Horodecki, Ryszard and Horodecki, Pawe\l{} and Horodecki, Micha\l{} and Horodecki, Karol},
  journal = {Rev. Mod. Phys.},
  volume = {81},
  issue = {2},
  pages = {865--942},
  numpages = {0},
  year = {2009},
  month = {Jun},
  publisher = {American Physical Society},
  doi = {10.1103/RevModPhys.81.865},
  url = {https://link.aps.org/doi/10.1103/RevModPhys.81.865}
}

@article{GURVITS2003sepNPhard,
title = {Classical complexity and quantum entanglement},
journal = {Journal of Computer and System Sciences},
volume = {69},
number = {3},
pages = {448-484},
year = {2004},
note = {Special Issue on STOC 2003},
issn = {0022-0000},
doi = {https://doi.org/10.1016/j.jcss.2004.06.003},
url = {https://www.sciencedirect.com/science/article/pii/S0022000004000893},
author = {Leonid Gurvits}
}

@article{EntDetRev_Guhne2009,
   title={Entanglement detection},
   volume={474},
   ISSN={0370-1573},
   url={http://dx.doi.org/10.1016/j.physrep.2009.02.004},
   DOI={10.1016/j.physrep.2009.02.004},
   number={1–6},
   journal={Physics Reports},
   publisher={Elsevier BV},
   author={Gühne, Otfried and Tóth, Géza},
   year={2009},
   month=apr, pages={1–75} }

@article{horodecki2001separability,
  title={Separability of n-particle mixed states: necessary and sufficient conditions in terms of linear maps},
  author={Horodecki, Micha{\l} and Horodecki, Pawe{\l} and Horodecki, Ryszard},
  journal={Physics Letters A},
  volume={283},
  number={1-2},
  pages={1--7},
  year={2001},
  publisher={Elsevier}
}

@article{doherty2004complete,
  title={Complete family of separability criteria},
  author={Doherty, Andrew C and Parrilo, Pablo A and Spedalieri, Federico M},
  journal={Phys. Rev. A},
  volume={69},
  number={2},
  pages={022308},
  year={2004},
  publisher={APS},
  doi={https://doi.org/10.1103/PhysRevA.69.022308}
}

@article{Huber2022DimFree,
	doi = {10.1007/s00220-022-04485-9},
	url = {https://doi.org/10.1007\%2Fs00220-022-04485-9},
	year = 2022,
	month = {aug},
        volume = {396},
        issue = {3},
        pages = {1051-1070},
	publisher = {Springer Science and Business Media {LLC}},
	author = {Felix Huber and Igor Klep and Victor Magron and Jurij Vol{\v{c}}i{\v{c}}},
	title = {Dimension-Free Entanglement Detection in Multipartite {W}erner States},
	journal = {Communications in Mathematical Physics}
}

@article{Huber2020PositiveMA,
doi = {10.1063/5.0028856},
url = {https://doi.org/10.1063\%2F5.0028856},
year = 2021,
month = {feb},
publisher = {{AIP} Publishing},
volume = {62},
number = {2},
pages = {022203}, 
author = {Felix Huber},
title = {Positive maps and trace polynomials from the symmetric group},
journal = {Journal of Mathematical Physics}
}

@article{MaassenSlides,
    title = {Entanglement of symmetric {W}erner states},
    year = {2019},
    author = {Maassen, H. and  K\"{u}mmerer, B.},
    journal = {Workshop: Mathematics of Quantum Information Theory, \url{http://www.bjadres.nl/MathQuantWorkshop/Slides/SymmWernerHandout.pdf}}
}

@article{Huber2021MatrixFO,
  title={Matrix forms of immanant inequalities},
  author={Felix Huber and Hans Maassen},
  url = {https://arxiv.org/abs/2103.04317},
  journal = {arXiv:2103.04317},
  year={2021}
}

@article{Elben2020,
  title = {Mixed-State Entanglement from Local Randomized Measurements},
  author = {Elben, Andreas and Kueng, Richard and Huang, Hsin-Yuan (Robert) and van Bijnen, Rick and Kokail, Christian and Dalmonte, Marcello and Calabrese, Pasquale and Kraus, Barbara and Preskill, John and Zoller, Peter and Vermersch, Benoit},
  journal = {Physical Review Letters},
  volume = {125},
  issue = {20},
  pages = {200501},
  numpages = {6},
  year = {2020},
  month = {Nov},
  publisher = {American Physical Society},
  doi = {10.1103/PhysRevLett.125.200501}
}

@article{Elben22Toolbox,
	doi = {10.1038/s42254-022-00535-2},
	url = {https://doi.org/10.1038\%2Fs42254-022-00535-2},
	year = 2022,
	month = {dec},
	publisher = {Springer Science and Business Media {LLC}},
	volume = {5},
	number = {1},
	pages = {9--24},
	author = {Andreas Elben and Steven T. Flammia and Hsin-Yuan Huang and Richard Kueng and John Preskill and Benoit Vermersch and Peter Zoller},
	title = {The randomized measurement toolbox},
	journal = {Nature Reviews Physics}
}

@article{Neven_SymmetryResMomentsPT2021,
	doi = {10.1038/s41534-021-00487-y},
	url = {https://doi.org/10.1038\%2Fs41534-021-00487-y},
	year = 2021,
	month = {oct},
	publisher = {Springer Science and Business Media {LLC}},
	volume = {7},
	number = {1},
	author = {Antoine Neven and Jose Carrasco and Vittorio Vitale and Christian Kokail and Andreas Elben and Marcello Dalmonte and Pasquale Calabrese and Peter Zoller and Benoit Vermersch and Richard Kueng and Barbara Kraus},
	title = {Symmetry-resolved entanglement detection using partial transpose moments},
	journal = {npj Quantum Information}
}

@article{Keyl-WernerEstimatingSpectra_2001,
  title = {Estimating the spectrum of a density operator},
  author = {Keyl, M. and Werner, R. F.},
  journal = {Physical Review A},
  volume = {64},
  issue = {5},
  pages = {052311},
  numpages = {5},
  year = {2001},
  publisher = {American Physical Society},
  doi = {10.1103/PhysRevA.64.052311},
  url = {https://link.aps.org/doi/10.1103/PhysRevA.64.052311}
}

@article{HoroRedCrit_1999,
  title = {Reduction criterion of separability and limits for a class of distillation protocols},
  author = {Horodecki, Micha\l{} and Horodecki, Pawe\l{}},
  journal = {Phys. Rev. A},
  volume = {59},
  issue = {6},
  pages = {4206--4216},
  numpages = {0},
  year = {1999},
  month = {Jun},
  publisher = {American Physical Society},
  doi = {10.1103/PhysRevA.59.4206},
  url = {https://link.aps.org/doi/10.1103/PhysRevA.59.4206}
}

@article{ImaiBoundfromRM_2021,
  title = {Bound Entanglement from Randomized Measurements},
  author = {Imai, Satoya and Wyderka, Nikolai and Ketterer, Andreas and G\"uhne, Otfried},
  journal = {Phys. Rev. Lett.},
  volume = {126},
  issue = {15},
  pages = {150501},
  numpages = {6},
  year = {2021},
  month = {Apr},
  publisher = {American Physical Society},
  doi = {10.1103/PhysRevLett.126.150501},
  url = {https://link.aps.org/doi/10.1103/PhysRevLett.126.150501}
}

@article{lew1966generalized,
  title={The generalized Cayley-Hamilton theorem in n dimensions},
  author={Lew, John S},
  journal={Zeitschrift f{\"u}r angewandte Mathematik und Physik ZAMP},
  volume={17},
  pages={650--653},
  year={1966},
  publisher={Springer},
  url={https://link.springer.com/article/10.1007/BF01597249#citeas}
}

@article{ginibre1965,
  title={Statistical ensembles of complex, quaternion, and real matrices},
  author={Ginibre, Jean},
  journal={J. Math. Phys.},
  volume={6},
  number={3},
  pages={440},
  year={1965},
  publisher={American Institute of Physics}
}

@book{marshall1979inequalities,
  title={Inequalities: {T}heory of {M}ajorization and its {A}pplications},
  author={Marshall, Albert W and Olkin, Ingram and Arnold, Barry C},
  year={1979},
  publisher={Springer}
}

@book{GeoQuantStates2006,
  title={Geometry of Quantum States: an Introduction to Quantum Entanglement},
  author={Bengtsson, Ingemar and {\.Z}yczkowski, Karol},
  year={2006},
  publisher={Cambridge University Press},
  doi={
    https://doi.org/10.1017/CBO9780511535048}
}

@article{choi1975completely,
  title={Completely positive linear maps on complex matrices},
  author={Choi, Man-Duen},
  journal={Linear Algebra Appl.},
  volume={10},
  number={3},
  pages={285--290},
  year={1975},
  publisher={Elsevier},
  doi={https://doi.org/10.1016/0024-3795(75)90075-0}
}

@article{jamiolkowski1972linear,
  title={Linear transformations which preserve trace and positive semidefiniteness of operators},
  author={Jamio{\l}kowski, Andrzej},
  journal={Rep. Math. Phys.},
  volume={3},
  number={4},
  pages={275--278},
  year={1972},
  publisher={Elsevier},
  doi={https://doi.org/10.1016/0034-4877(72)90011-0}
}

@book{peres1997quantum,
  title={Quantum Theory: Concepts and Methods},
  author={Peres, Asher},
  volume={72},
  year={1998},
  publisher={Kluver Academic, London},
  doi={https://doi.org/10.1007/0-306-47120-5}
}

@article{HorodeckiBoundEnt_1998,
   title={Mixed-State Entanglement and Distillation: Is there a “Bound” Entanglement in Nature?},
   volume={80},
   ISSN={1079-7114},
   url={http://dx.doi.org/10.1103/PhysRevLett.80.5239},
   DOI={10.1103/physrevlett.80.5239},
   number={24},
   journal={Physical Review Letters},
   publisher={American Physical Society (APS)},
   author={Horodecki, Michał and Horodecki, Paweł and Horodecki, Ryszard},
   year={1998},
   month=jun, pages={5239–5242} }

@book{bhatia2013matrix,
  title={Matrix analysis},
  author={Bhatia, Rajendra},
  volume={169},
  year={2013},
  publisher={Springer Science \& Business Media},
  doi={https://doi.org/10.1007/978-1-4612-0653-8}
}

@article{PmapsWBAlg_Miqueleta2024,
doi = {10.1088/1751-8121/ad2b86},
url = {https://dx.doi.org/10.1088/1751-8121/ad2b86},
year = {2024},
month = {mar},
publisher = {IOP Publishing},
volume = {57},
number = {11},
pages = {115202},
author = {Maria Balanzó-Juandó and Michał Studziński and Felix Huber},
title = {Positive maps from the walled Brauer algebra},
journal = {Journal of Physics A: Mathematical and Theoretical}
}

@article{OptEW_Lewenstein2000,
  title = {Optimization of entanglement witnesses},
  author = {Lewenstein, M. and Kraus, B. and Cirac, J. I. and Horodecki, P.},
  journal = {Phys. Rev. A},
  volume = {62},
  issue = {5},
  pages = {052310},
  numpages = {16},
  year = {2000},
  month = {Oct},
  publisher = {American Physical Society},
  doi = {10.1103/PhysRevA.62.052310},
  url = {https://link.aps.org/doi/10.1103/PhysRevA.62.052310}
}

@article{4partUnlockBound_Smolin2001,
  title = {Four-party unlockable bound entangled state},
  author = {Smolin, John A.},
  journal = {Phys. Rev. A},
  volume = {63},
  issue = {3},
  pages = {032306},
  numpages = {4},
  year = {2001},
  month = {Feb},
  publisher = {American Physical Society},
  doi = {10.1103/PhysRevA.63.032306},
  url = {https://link.aps.org/doi/10.1103/PhysRevA.63.032306}
}

@article{BellCorActBoundEnt_Bandy2005,
  title = {Bell-correlated activable bound entanglement in multiqubit systems},
  author = {Bandyopadhyay, Somshubhro and Chattopadhyay, Indrani and Roychowdhury, Vwani and Sarkar, Debasis},
  journal = {Phys. Rev. A},
  volume = {71},
  issue = {6},
  pages = {062317},
  numpages = {5},
  year = {2005},
  month = {Jun},
  publisher = {American Physical Society},
  doi = {10.1103/PhysRevA.71.062317},
  url = {https://link.aps.org/doi/10.1103/PhysRevA.71.062317}
}

@article{RemoteInfoConcUnlockBoundApplic_Murao2001,
  title = {Remote Information Concentration Using a Bound Entangled State},
  author = {Murao, Mio and Vedral, Vlatko},
  journal = {Phys. Rev. Lett.},
  volume = {86},
  issue = {2},
  pages = {352--355},
  numpages = {0},
  year = {2001},
  month = {Jan},
  publisher = {American Physical Society},
  doi = {10.1103/PhysRevLett.86.352},
  url = {https://link.aps.org/doi/10.1103/PhysRevLett.86.352}
}

@article{BoundMaxViolBellUnlockApplic_Augusiak2006,
  title = {Bound entanglement maximally violating Bell inequalities: Quantum entanglement is not fully equivalent to cryptographic security},
  author = {Augusiak, Remigiusz and Horodecki, Pawel},
  journal = {Phys. Rev. A},
  volume = {74},
  issue = {1},
  pages = {010305},
  numpages = {4},
  year = {2006},
  month = {Jul},
  publisher = {American Physical Society},
  doi = {10.1103/PhysRevA.74.010305},
  url = {https://link.aps.org/doi/10.1103/PhysRevA.74.010305}
}

@article{QComComplexBoundUnlockApplic_Bruckner2002,
  title = {Quantum Communication Complexity Protocol with Two Entangled Qutrits},
  author = {Brukner, \ifmmode \check{C}\else \v{C}\fi{}aslav and \ifmmode \dot{Z}\else \.{Z}\fi{}ukowski, Marek and Zeilinger, Anton},
  journal = {Phys. Rev. Lett.},
  volume = {89},
  issue = {19},
  pages = {197901},
  numpages = {4},
  year = {2002},
  month = {Oct},
  publisher = {American Physical Society},
  doi = {10.1103/PhysRevLett.89.197901},
  url = {https://link.aps.org/doi/10.1103/PhysRevLett.89.197901}
}

@article{SepDistMultiPartBoundUnlock_Dur1999,
  title = {Separability and Distillability of Multiparticle Quantum Systems},
  author = {D\"ur, W. and Cirac, J. I. and Tarrach, R.},
  journal = {Phys. Rev. Lett.},
  volume = {83},
  issue = {17},
  pages = {3562--3565},
  numpages = {0},
  year = {1999},
  month = {Oct},
  publisher = {American Physical Society},
  doi = {10.1103/PhysRevLett.83.3562},
  url = {https://link.aps.org/doi/10.1103/PhysRevLett.83.3562}
}

@misc{clarisse2006PhDDistillation,
      title={Entanglement Distillation; A Discourse on Bound Entanglement in Quantum Information Theory}, 
      author={Lieven Clarisse},
      year={2006},
      eprint={quant-ph/0612072},
      archivePrefix={arXiv},
      primaryClass={quant-ph},
      url={https://arxiv.org/abs/quant-ph/0612072}, 
}

@article{advancesHDent_Erhard2020,
  title={Advances in high-dimensional quantum entanglement},
  author={Erhard, Manuel and Krenn, Mario and Zeilinger, Anton},
  journal={Nature Reviews Physics},
  volume={2},
  number={7},
  pages={365--381},
  year={2020},
  publisher={Nature Publishing Group UK London}
}

@article{MPEntSuperQubits_Lu2022,
  title = {Multipartite Entanglement in Rabi-Driven Superconducting Qubits},
  author = {Lu, Marie and Ville, Jean-Loup and Cohen, Joachim and Petrescu, Alexandru and Schreppler, Sydney and Chen, Larry and J\"unger, Christian and Pelletti, Chiara and Marchenkov, Alexei and Banerjee, Archan and Livingston, William P. and Kreikebaum, John Mark and Santiago, David I. and Blais, Alexandre and Siddiqi, Irfan},
  journal = {PRX Quantum},
  volume = {3},
  issue = {4},
  pages = {040322},
  numpages = {15},
  year = {2022},
  month = {Nov},
  publisher = {American Physical Society},
  doi = {10.1103/PRXQuantum.3.040322},
  url = {https://link.aps.org/doi/10.1103/PRXQuantum.3.040322}
}

@article{DetMPentManyBody_Frerot2022,
  title = {Unveiling Quantum Entanglement in Many-Body Systems from Partial Information},
  author = {Fr\'erot, Ir\'en\'ee and Baccari, Flavio and Ac\'{\i}n, Antonio},
  journal = {PRX Quantum},
  volume = {3},
  issue = {1},
  pages = {010342},
  numpages = {19},
  year = {2022},
  month = {Mar},
  publisher = {American Physical Society},
  doi = {10.1103/PRXQuantum.3.010342},
  url = {https://link.aps.org/doi/10.1103/PRXQuantum.3.010342}
}

@inproceedings{LargScaleQNetworksKozlowski_2019, 
    series={NANOCOM ’19},
   title={Towards Large-Scale Quantum Networks},
   url={http://dx.doi.org/10.1145/3345312.3345497},
   DOI={10.1145/3345312.3345497},
   booktitle={Proceedings of the Sixth Annual ACM International Conference on Nanoscale Computing and Communication},
   publisher={ACM},
   author={Kozlowski, Wojciech and Wehner, Stephanie},
   year={2019},
   month=sep, collection={NANOCOM ’19} }

@article{EfficientLargeScaleMBdyn_Artaco2024,
  title = {Efficient Large-Scale Many-Body Quantum Dynamics via Local-Information Time Evolution},
  author = {Artiaco, Claudia and Fleckenstein, Christoph and Aceituno Ch\'avez, David and Kvorning, Thomas Klein and Bardarson, Jens H.},
  journal = {PRX Quantum},
  volume = {5},
  issue = {2},
  pages = {020352},
  numpages = {28},
  year = {2024},
  month = {Jun},
  publisher = {American Physical Society},
  doi = {10.1103/PRXQuantum.5.020352},
  url = {https://link.aps.org/doi/10.1103/PRXQuantum.5.020352}
}

@article{DetMBodyPermuMoments_Liu2022,
  title = {Detecting Entanglement in Quantum Many-Body Systems via Permutation Moments},
  author = {Liu, Zhenhuan and Tang, Yifan and Dai, Hao and Liu, Pengyu and Chen, Shu and Ma, Xiongfeng},
  journal = {Phys. Rev. Lett.},
  volume = {129},
  issue = {26},
  pages = {260501},
  numpages = {6},
  year = {2022},
  month = {Dec},
  publisher = {American Physical Society},
  doi = {10.1103/PhysRevLett.129.260501},
  url = {https://link.aps.org/doi/10.1103/PhysRevLett.129.260501}
}

@article{DetMBodyContinuous_Kunkel2022,
  title = {Detecting Entanglement Structure in Continuous Many-Body Quantum Systems},
  author = {Kunkel, Philipp and Pr\"ufer, Maximilian and Lannig, Stefan and Strohmaier, Robin and G\"arttner, Martin and Strobel, Helmut and Oberthaler, Markus K.},
  journal = {Phys. Rev. Lett.},
  volume = {128},
  issue = {2},
  pages = {020402},
  numpages = {5},
  year = {2022},
  month = {Jan},
  publisher = {American Physical Society},
  doi = {10.1103/PhysRevLett.128.020402},
  url = {https://link.aps.org/doi/10.1103/PhysRevLett.128.020402}
}

@article{NetworkGMPE_Navascues2020,
  title = {Genuine Network Multipartite Entanglement},
  author = {Navascu\'es, Miguel and Wolfe, Elie and Rosset, Denis and Pozas-Kerstjens, Alejandro},
  journal = {Phys. Rev. Lett.},
  volume = {125},
  issue = {24},
  pages = {240505},
  numpages = {6},
  year = {2020},
  month = {Dec},
  publisher = {American Physical Society},
  doi = {10.1103/PhysRevLett.125.240505},
  url = {https://link.aps.org/doi/10.1103/PhysRevLett.125.240505}
}

@article{liu2022MomentsPermutations,
  title={Detecting entanglement in quantum many-body systems via permutation moments},
  author={Liu, Zhenhuan and Tang, Yifan and Dai, Hao and Liu, Pengyu and Chen, Shu and Ma, Xiongfeng},
  journal={Physical Review Letters},
  volume={129},
  number={26},
  pages={260501},
  year={2022},
  publisher={APS}
}

@article{CollectiveRM_Imai2024,
  title = {Collective Randomized Measurements in Quantum Information Processing},
  author = {Imai, Satoya and T\'oth, G\'eza and G\"uhne, Otfried},
  journal = {Phys. Rev. Lett.},
  volume = {133},
  issue = {6},
  pages = {060203},
  numpages = {8},
  year = {2024},
  month = {Aug},
  publisher = {American Physical Society},
  doi = {10.1103/PhysRevLett.133.060203},
  url = {https://link.aps.org/doi/10.1103/PhysRevLett.133.060203}
}

@misc{zhang2023experimentalverificationboundmultiparticle,
      title={Experimental verification of bound and multiparticle entanglement with the randomized measurement toolbox}, 
      author={Chao Zhang and Yuan-Yuan Zhao and Nikolai Wyderka and Satoya Imai and Andreas Ketterer and Ning-Ning Wang and Kai Xu and Keren Li and Bi-Heng Liu and Yun-Feng Huang and Chuan-Feng Li and Guang-Can Guo and Otfried Gühne},
      year={2023},
      eprint={2307.04382},
      archivePrefix={arXiv},
      primaryClass={quant-ph},
      url={https://arxiv.org/abs/2307.04382}, 
}

@article{QSSMultiBoundEnt_Zhou2018,
  title = {Quantum Secret Sharing Among Four Players Using Multipartite Bound Entanglement of an Optical Field},
  author = {Zhou, Yaoyao and Yu, Juan and Yan, Zhihui and Jia, Xiaojun and Zhang, Jing and Xie, Changde and Peng, Kunchi},
  journal = {Phys. Rev. Lett.},
  volume = {121},
  issue = {15},
  pages = {150502},
  numpages = {6},
  year = {2018},
  month = {Oct},
  publisher = {American Physical Society},
  doi = {10.1103/PhysRevLett.121.150502},
  url = {https://link.aps.org/doi/10.1103/PhysRevLett.121.150502}
}

@article{Lieb_Permanent_Conj_1996,
 ISSN = {00959057, 19435274},
 URL = {http://www.jstor.org/stable/24901474},
 author = {ELLIOTT H. LIEB},
 journal = {Journal of Mathematics and Mechanics},
 number = {2},
 pages = {127--134},
 publisher = {Indiana University Mathematics Department},
 title = {Proofs of some Conjectures on Permanents},
 urldate = {2024-10-29},
 volume = {16},
 year = {1966}
}

@article{Code-data,
  title = {{AR}-open-access},
  author = {Albert Rico},
  year = {2025},
  journal = {\url{https://github.com/AlbertRico/AR-open-access/tree/main}},
  url = {https://github.com/AlbertRico/AR-open-access/tree/main/arXiv%3A2502.17697}
}

\appendix
\section{Examples of Young projectors}\label{App:YoungProjectors}
Consider the Hilbert space of $k=3$ qudits, $\mathcal{H}=\C^d\otimes\C^d\otimes\C^d$. The symmetric group $S_3$ has three invariant subspaces: the symmetric, standard, and antisymmetric subspace. These correspond respectively to the partitions $(3,0,0)$, $(2,1,0)$ and $(1,1,1)$, and are defined by the following Young projectors,
\begin{align}
    P_{3}&=\frac{1}{6}\eta_d\Big (\big [(\id)+(12)+(13)+(23)+(123)+(132)\big ]\Big )\,, \\
    P_{2,1}&= \frac{1}{6} \eta_d\Big (\big [2(\id)-(123)-(132)\big ]\Big )\,, \\
    P_{1^3}&= \frac{1}{6}\eta_d\Big (\big [(\id)-(12)-(13)-(23)+(123)+(132)\big ]\Big )\,.
\end{align}
These have the property that $P_\lambda=0$ if $|\lambda|<d$~\cite{Keyl-WernerEstimatingSpectra_2001}. For example the three-qubit antisymmetrizer vanishes, namely
\begin{equation}
    P_{1,1,1}=0
\end{equation}
for $d=2$.

\section{Proof of Proposition~\ref{prop:GeneralMap}}\label{app:ProofMainProp}
\begin{proof}
Let us expand Eq.~\eqref{eq:PosMapMain}
into
\begin{align}\label{eq:PosMapMainExpanded}
    \Psi_{\Vec{a}}^{E}(\varrho_1,...,\varrho_{k-1})=&\sum_\lambda a_\lambda\sum_{\pi\in S_k}\chi_\lambda(\pi)E\prod_{r=2}^R\Big (\varrho_{\alpha^{(\pi)}_r}E\Big )\nonumber\\
        \tr\Big [&\prod_{s=1}^S\Big (\varrho_{\beta^{(\pi)}_s}E\Big )\Big ]\nonumber\\
        \cdots\tr\Big [&\prod_{t=1}^T\Big (\varrho_{\xi^{(\pi)}_t}E\Big )\Big ]\,,
\end{align}
where each permutation $\pi$ is written in its cycle form $\pi=(\alpha)(\beta)...(\xi)=(\alpha_1...\alpha_R)(\beta_1...\beta_S)...(\xi_1...\xi_T)$ with the convention $\alpha_1=1$. 

Let us now define the inner product $\langle u,w\rangle=\bra{u}X\ket{w}$ with $\ket{u},\ket{w}\in\mathbb{C}^{d}$ and write $G$ as a Gram matrix,
    \begin{equation}
    G=
    \begin{pmatrix}
        \bra{i_1}X\ket{i_1} & \cdots & \bra{i_1}X\ket{i_n} \\
        \vdots & \ddots &  \\
        \bra{i_n}X\ket{i_1} &  & \bra{i_n}X\ket{i_n} \\
    \end{pmatrix}\,.
    \end{equation}
By expanding each permutation 
$\pi\in S_k$ into cycles, $\pi=(\alpha)...(\xi)$, a linear combination of Immanants $\imm_\lambda$ which is positive on any matrix $G$ reads
    \begin{align}
        &\sum_\lambda a_\lambda\imm_\lambda (G)=\sum_\lambda a_\lambda\sum_{\pi\in S_k}\chi_\lambda(\pi)\prod_{l=1}^n\bra{i_l}X\ket{i_{\pi(l)}}=\nonumber\\
        &\sum_\lambda a_\lambda\sum_{\pi\in S_k}\chi_\lambda(\pi)\bra{i_{\alpha^{(\pi)}_{1}}}X\ket{i_{\alpha^{(\pi)}_{2}}}
        \cdots\bra{i_{\alpha^{(\pi)}_{R}}}X\ket{i_{\alpha^{(\pi)}_{1}}}\nonumber\\
        &\quad\quad\quad\quad\quad\quad\,\,...\nonumber\\
        &\,\quad\quad\quad\quad\quad\bra{i_{\xi^{(\pi)}_{1}}}X\ket{i_{\xi^{(\pi)}_{2}}}
        \cdots
        \bra{i_{\xi^{(\pi)}_{T}}}X\ket{i_{\xi^{(\pi)}_{1}}}\geq 0\,.
    \end{align}
    Let us now take a convex combination with nonnegative coefficients $x_i\geq 0$ and identify the positive semidefinite $d\times d$ matrices $\{\varrho_{\alpha^{(\pi)}_{v}}\}=\big \{\sum_{i_{\alpha^{(\pi)}_{r}=1}}^d x_{i_{\alpha^{(\pi)}_{r}}}|i_{\alpha^{(\pi)}_{r}}\rangle\langle i_{\alpha^{(\pi)}_{r}}|\big \}$ for the first cycle $\alpha$ of the permutation $\pi$, and similarly for the rest of cycles $\beta,...,\xi$ of $\pi$. 
    
    By assumption, recall that $\alpha^{(\pi)}_{1}=1$ in all terms of the summation. Therefore we can factorize $\langle i_{\alpha^{(\pi)}_{1}}|$ to the left hand side and $|i_{\alpha^{(\pi)}_{1}}\rangle$ to the right hand side, and take a linear combination of the whole expression in order to have it in terms of $\{\varrho_{\alpha^{(\pi)}_{r}}\}$,...,$\{\varrho_{\xi^{(\pi)}_{t}}\}$.
    It follows that
    \begin{align}
        \langle i_{\alpha^{(\pi)}_{1}}|\Psi_{\Vec{a}}^{E}(\varrho_1,...,\varrho_{k-1})| i_{\alpha^{(\pi)}_{1}}\rangle\nonumber
    \end{align}
    is nonnegative for any vector $| i_{\alpha^{(\pi)}_{1}}\rangle\in\mathbb{C}^{d}$ and hence the multilinear map $\Psi_{\Vec{a}}^{E}$ is positive, which proves the claim.
\end{proof}

\section{New positive maps from Proposition~\ref{prop:GeneralMap} for $n=3$}\label{app:ExsIneqsN=3}
Let us now consider the case $n=3$. From positivity of the determinant in $3\times 3$ positive semidefinite matrices, we obtain the positive map
\begin{align}\label{eq:PositiveMapsFamilyN=3}
\Psi_{(0,1)}^{E}(Y,Z)&=X\tr[YX]\tr[ZX]+XYXZX+XZXYX \nonumber\\
&-XYX\tr[ZX]-X\tr[YXZX]\nonumber\\
&-XZX\tr[YX]\geq 0  
\end{align}
for any three positive semidefinite matrices $X,Y,Z\geq 0$, which is linear in $Y$ and $Z$. By setting $Z=\one$, we have
\begin{align}\label{eq:PositiveMapsFamilyN=3C=Id}
    \Psi_{(0,1)}^{E}(Y,\one)&=X\tr[YX]\tr[X]+XYX^2+X^2YX\\
    &-XYX\tr[X]-X\tr[YX^2]-X^2\tr[YX]\geq 0\,.\nonumber
\end{align}
By construction, with $X=\one$ we recover the following map defined in~\cite{lew1966generalized,Huber2020PositiveMA},
\begin{align}\label{eq:Cayley-hamiltoniand=3}
    \Psi_\one(Y,Z)&\one\tr[Y]\tr[Z]+YZ+ZY\\
    &-Y\tr[Z]-\one\tr[YZ]-Z\tr[Y]\geq 0\,.\nonumber
\end{align}
If in equation~\eqref{eq:PositiveMapsFamilyN=3} we choose $C=A=\one$ (or similarly $B=A=\one$), we are left with the reduction map up to a rescaling,
\begin{equation}
\begin{aligned}
    \Psi_\one(Y,\one)=(d-2)\one\tr[Y]-(d-2)Y\geq 0,
\end{aligned}
\end{equation}
where $d$ is the size of the positive matrix $Y$.

\section{Multipartite entanglement from a single immanant inequality}\label{app:GenMultiMapsWits}
Here we will show how to obtain positive maps and witnesses to detect multipartite entanglement from an immanant inequality, generalizing the connections established in~\cite{MaassenSlides,Huber2021MatrixFO}.

\begin{cor}\label{cor:PosMapYoungProj}
Let $\varrho$ be a quantum state shared among $(n-1)$ parties with local dimension $d$. Let $\sum_{\lambda\vdash n}a_\lambda\imm_\lambda(G)\geq 0$ be an immanant inequality for any $n\times n$ positive matrix $G\geq 0$. If $\varrho$ is fully separable, then
\begin{equation}\label{eq:Method2}
    \tr_{1,...,n\setminus 1}\Bigg [\sum_{\lambda\vdash n} a_\lambda P_\lambda(E^{\otimes n})(\one_d\otimes\varrho)\Bigg ]\geq 0
\end{equation}
for any $d\times d$ positive semidefinite matrix $X\geq 0$.
\end{cor}
\begin{proof}
The proof follows from Proposition~\ref{prop:GeneralMap} by decomposing a fully separable state $\varrho$ as
\begin{equation}
    \varrho=\sum_i p_i \varrho_1^i\otimes\cdots\otimes \varrho_{k-1}^i
\end{equation}
where $p_i\geq 0$, $\sum_i p_i=1$ and $\varrho_l^i$ are quantum states for all $i$ and $l$. It is now clear that Eq.~\eqref{eq:Method2} is a convex combination of positive operators.
\end{proof}
This allows to extend the linear techniques for entanglement detection derived in~\cite{MaassenSlides,Huber2020PositiveMA,Huber2021MatrixFO} by adding a local operator $E^{\otimes n}$: 
Using that an operator $M$ is positive semidefinite if and only if $\tr(MN)\geq 0$ for all positive semidefinite operators $N\geq 0$; and that $\tr\big (\tr_1 (M) N\big ) = \tr\big (M(\one\otimes N)\big )$, 
Corollary~\ref{cor:PosMapYoungProj} implies that if an $n$-partite state $\varrho$ acting on $\cdn$ is fully separable and $\sum_{\lambda\vdash n}a_\lambda\imm_\lambda(G)\geq 0$ is an immanant inequality, then the operator
\begin{equation}\label{eq:WitMaassenApp}
 W = E^{\otimes n}\sum_{\lambda\vdash n}\frac{a_\lambda}{\Xi_\lambda(\id)}P_{\lambda}
\end{equation}
has a nonnegative expectation value on $\varrho$, $\tr(W\varrho)\geq 0$. In particular $W$ is a witness if it is not positive semidefinite. Physically, this can be thought of all local parties agreeing on performing the same local filter $F$ with $FF^\dag=X$ and then evaluating a witness using the constructions of~\cite{MaassenSlides}, which are recovered when the filter $F$ is unitary.

As an example for bipartite states, let us choose the inequality $\prod_iG_{ii}-\det(G)\geq 0$ for $3\times 3$ positive semidefinite matrices $G\geq 0$. By Corollary~\ref{cor:PosMapYoungProj}, 
\begin{equation}
    \tr_{AB}((\one-6P_{1,1,1})(X^{\otimes 3})(\varrho_{AB}\otimes\one_C))\geq 0
\end{equation}
if the state $\varrho_{AB}$ is separable. In terms of expectation values, this statement for bipartite states means that for three-partite systems,
\begin{equation}
    W = (\one-6P_{1,1,1})X^{\otimes 3}
\end{equation}
is an entanglement witness for any positive operator $X\geq 0$. Note that this recovers the known witness $\one-6P_{1,1,1}$~\cite{MaassenSlides,TP_Rico24} for $X=\one$. 

\section{Proof of Proposition~\ref{prop:WitPartContr}}\label{app:ProofPropWitnesses}
\begin{proof}
To construct a new class of witnesses to detect $k\times(n-1)$-partite entangled states using the technique in~\cite{short}, we will employ a set of $k$ $n$-partite witnesses or positive operators constructed from immanant inequalities according to~\eqref{eq:WitMaassen}~\cite{MaassenSlides},
\begin{equation}
 W^{(i)}:=\sum_{\lambda\vdash n}a^{(i)}_\lambda P_\lambda\,,
\end{equation}
where each witness (or positive operator) $W^{(i)}$ is labeled by $i = 1,...,k$.
Consider their tensor product,
\begin{equation}
 O = W^{(1)}_{1,...,n}\otimes W^{(2)}_{n+1,...,2n}\otimes...\otimes W^{(k)}_{n(k-1),...,nk}\,.
\end{equation}
Now let $\varrho$ be a $k\times(n-1)$-partite state. It follows from Corollary~\ref{cor:PosMapYoungProj} that if $\varrho$ is fully separable, then
\begin{equation}\label{eq:PosMapnxkStates}
 \tr_{1...kn\setminus \{in\}_{i=1}^k}  \Big (O\cdot\varrho\otimes \one_{n,2n,...,kn}\Big )\geq 0
\end{equation}
is positive semidefinite. 

We will use that an operator $X$ is positive semidefinite, $X\geq 0$, if and only if its inner product with any positive semiedefinite operator $Y\geq 0$ is nonnegative,
\begin{equation}
 \tr(XY)\geq 0.
\end{equation}
Therefore, positivity of Eq.~\eqref{eq:PosMapnxkStates} is equivalent to the following equation being nonnegative for any $k$-partite positive operator $\tau$,
\begin{equation}\label{eq:PosMapnxkStatesTrace}
 \tr  \Big (O\cdot\varrho\otimes \tau_{n,2n,...,kn}\Big )\geq 0\,.
\end{equation}
Using the cooridnate-free definition of the partial trace,
\begin{equation}
 \tr(M(N\otimes\one))=\tr(\tr_1(M)N)\,,
\end{equation}
Eq.~\eqref{eq:PosMapnxkStatesTrace} can be written as
\begin{equation}
 \tr\Big (\tr_{\{in\}_{i=1}^k}\big (O(\one\otimes\tau_{\{in\}_{i=1}^k})\big )\cdot \varrho_{1,...,kn\setminus \{in\}_{i=1}^k}\Big )\geq 0\,.
\end{equation}
This shows the claim.
\end{proof}

\end{document}